\documentclass[12pt]{article}
\usepackage{amsthm,amsfonts,amssymb,amsmath}
\usepackage{xcolor}
\usepackage{graphicx}
\usepackage{subcaption}

\usepackage{xspace}
\usepackage{xfrac}

\newcommand{\e}{\mathrm{e}}
\newcommand{\bM}{\mathbf{M}}
\newcommand{\bQ}{\mathbf{Q}}

\newcommand{\be}{\mathbf{e}}
\newcommand{\bz}{\mathbf{z}}
\newcommand{\bw}{\mathbf{w}}
\newcommand{\bv}{\mathbf{v}}
\newcommand{\bu}{\mathbf{u}}
\newcommand{\bp}{\mathbf{p}}
\newcommand{\bq}{\mathbf{q}}
\newcommand{\bi}{\mathbf{i}}
\newcommand{\bx}{\mathbf{x}}

\newcommand{\bone}{\mathbf{1}}
\newcommand{\bF}{\mathbf{F}}
\newcommand{\rp}{\mathrm{p}}
\newcommand{\rd}{\mathrm{d}}

\newcommand{\A}{\ensuremath{\mathbb{A}}\xspace}
\newcommand{\B}{\ensuremath{\mathbb{B}}\xspace}
\newcommand{\tsp}{\mathsf{s}}
\newcommand{\razaoP}{\frac{P_{\overrightarrow{C}}}{P_{\overleftarrow{C}}}}

\newcommand{\R}{\mathbb{R}}
\newcommand{\ini}{\mathrm{I}}

\newcommand{\simplex}{\Delta}
\newcommand{\bydef}{:=}

\newtheorem{lemma}{Lemma}
\newtheorem{cor}{Corollary}

\newtheorem{remark}{Remark}
\usepackage{authblk}
\author[1]{Diogo Costa-Cabanas}
\author[1,2]{Fabio A. C. C. Chalub\footnote{Corresponding author:\texttt{facc@fct.unl.pt}}}
\author[3]{Max O. Souza}
\affil[1]{Departamento de Matem\'atica, Faculdade de Ci\^encias e Tecnologia, Universidade Nova de Lisboa, Quinta da Torre, 2829-516 Caparica, Portugal}
\affil[2]{Centro de Matem\'atica e Aplica\c c\~oes, Faculdade de Ci\^encias e Tecnologia, Universidade Nova de Lisboa, Quinta da Torre, 2829-516 Caparica, Portugal}
\affil[3]{Instituto de Matem\'atica e Estatística, Universidade Federal Fluminense, Rua Prof. Marcos Waldemar de Freitas Reis, S/N, Campus do Gragoat\'a, Niter\'oi, RJ 24210-201, Brazil}

\title{Entropy and the arrow of time in population dynamics}
\date{\today}

\usepackage[deletedmarkup=sout]{changes}
\definechangesauthor[name=Fabio, color=blue]{facc}
\definechangesauthor[name=Diogo, color=red]{dcc}
\definechangesauthor[name=Max, color=magenta]{mos}

\usepackage{todonotes}
  
\begin{document}

\maketitle

\begin{abstract}
The concept of entropy in statistical physics is related to the existence of irreversible macroscopic processes.  In this work, we explore  a recently introduced entropy formula for a class of stochastic processes with more than one absorbing state that is extensively used in population genetics models. We will consider the Moran process as a paradigm for this class, and will extend our discussion to other models outside this class. We will also discuss the relation between non-extensive entropies in physics and epistasis (i.e., when the effects of different alleles are not independent) and the role of symmetries in population genetic models.

\end{abstract}

\textbf{Keywords:} Entropy, Moran process, time-irreversible processes, epistasis, fundamental symmetries.

\textbf{PhySH:} Evolutionary dynamics, Population dynamics 
\section{Introduction}

\subsection{Background}

Time reversal symmetry is one of the fundamental symmetry in physical laws; as a simple, but clear example, second Newton's law is preserved by time reversal $t\mapsto -t$. In classical Hamiltonian mechanics, the dynamics of a given system is encoded in a function $H(p,q,t)$, where $p$ and $q$ are the generalised position and momentum, respectively, of the corresponding particles and the time $t$ --- this function is called the Hamiltonian of the system. If the time dependence is not explicit, i.e., if $H(p,q,t)=H(p,q)$, then it corresponds to the total energy of the system. In the latter case, it is often assumed that  $H(p,q)=H(p,-q)$ and, therefore $(q(t),p(t))$ solves Hamilton's equations if and only if $(q(-t),-p(-t))$ is also a solution.

However, at human scale (i.e., when considering the number of interacting agents in a given system at the order of $10^{23}$ particles), physics is full of irreversible phenomena. A clear example of this situation happens when a drop of ink dissolves in a bucket of water. Nothing prevents a spontaneous concentration of ink from a previously homogeneous mixture; however, these phenomena are expected to happen --- even with almost negligible probability -- after an interval of time larger than the age of the universe.

The irreversibility of a certain class of physical phenomena is natural only in the realm of statistical mechanics --- the area of physics that deals with a large number of interacting constituents. One of the aims of the present work is to use techniques from statistical mechanics to understand irreversible phenomena in models used in population genetics.

Population genetics has no equivalent to the second Newton's law. Furthermore, most models based on microscopic descriptions of a population (e.g., the Moran, and the Wright-Fisher processes, Individual-based dynamics, to name a few models used in the study of biological evolution) are first-order in time, and therefore it does not possess the symmetry $t \mapsto -t$. Therefore, it is not unexpected that population based mathematical models do not present, generally, the time-reversal symmetry. Examples of such models are the replicator dynamics~\cite{HofbauerSigmund}, the Kimura equation~\cite{ChalubSouza_JMB2014}, the canonical equation of adaptive dynamics~\cite{Champagnat_2002} etc.

On the other hand, both classical particle physics and population genetics starts with the description of the dynamics at individual level. However,  the relevant features are measured in completely different scales -- in fact, at population level, although this expression is not used in physics. 

In this work, we propose to explore one of the central concepts of the micro-macro asymmetry in physics, the \emph{entropy}, in the context of evolutionary dynamics. We will use the term as close as possible to its meaning in statistical mechanics.

The concept of entropy was introduced in physics within the framework of the study of efficiency in thermal machines; later on, Boltzmann reinterpreted this concept as a measure of the number of microstates consistent with a given macroscopic state of a system with a large number of degrees of freedom. The understanding of entropy and the associated second law of thermodynamics is fundamental to understand the asymmetry between past and future --- the so called \emph{arrow of time}, cf.~\cite{Feynman_1,Lebowitz_1999}.

The implications of the concept of entropy went far beyond physics; in a further development, Shannon extended Boltzmann ideas to what is now called \emph{information theory}~\cite{Shannon_Weaver}. Currently, many reinterpretations of this concept are studied in general biology, including~\cite{wiley1982victims,brooks1984evolution,collier1986entropy,wicken1983entropy}.

Here, we are concerned with understanding  non-equilibrium dynamics, and, in particular,  irreversibility, in population dynamics. More precisely, we start this work by considering a population of interacting individuals, in which individuals are replaced by newborns over time, according to certain dynamics. These newborns inherit the characteristics from their parents.

In population dynamics, as discussed above, we are not concerned with the precise characterization of individuals, but with macroscopic descriptions, i.e., descriptions at population level. In particular, we study how allele frequency --- the fraction of the population that shares a given allele --- varies over time when the population evolves according to certain rules defined at individual level.

We claim that the entropies introduced in~\cite{CMRS:21} are the relevant quantities to characterize the irreversibility feature of evolutionary dynamics. More precisely, they form a class of functions obtained naturally from the mathematical theory of gradient flows and optimal transport; therefore, they are not only monotonic in time, but they increase optimally.  We will proceed with a detailed study of this concept for populations evolving according to the Moran process and discuss the application of the concept in models in which the mathematical theory does not apply directly.

\subsection{Outline}

In Section~\ref{sec:Moran}, we will discuss the basics of the Moran process (Subsection~\ref{ssec:basics}) and show that given the fitnesses of different types in a population, there is natural \emph{arrow of time} between different possible states (Subsection~\ref{ssec:irreversibility}). 

The properties of the entropy will be explored in  Section~\ref{sec:entropy}. In Subsection~\ref{ssec:decay}, we show that it is a monotonic decreasing function, stationary if and only if the population state is a linear combination of stationary and quasi-stationary states. We show, in the sequel (Subsection~\ref{ssec:asymptotic}) that the entropy decays exponentially, with the decay rate given by twice the second spectral gap of the Moran transition matrix, and linear coefficient directly related to the third eigenvector of the same matrix. In the final subsection, we discuss the relationship between different entropies, particularly the discussion between additive and subadditive entropies, and coevolution of different loci.

In Section~\ref{sec:beyond}, we go beyond the mathematical results discussed so far, and apply the theory, rigorously developed for the Moran process, to the more realistic Wright-Fisher process (Subsection~\ref{ssec:WF}), discuss a curious relation between minimum initial entropy, and minimum entropy in the long run (Subsection~\ref{ssec:eigenvector}) and in Subsection~\ref{sec:symmetries}, we speculate on the role of symmetries in population dynamics, discussing, in particular, why the replicator equation presents additional symmetries with respect to the models from which it is derived. Finally, we present some conclusions and speculate on possible biological applications to be addressed in the near future using the concepts discussed in this work.

\section{The Moran process}
\label{sec:Moran}

\subsection{Basic setup and notation}
\label{ssec:basics}

Consider a population of two-type individuals, \A and \B. We define the \emph{type selection vector} $(\tsp_0,\tsp_1,\dots,\tsp_N)$ such that $\tsp_i$ is the probability to select an individual of type \A in a population of fixed size $N$ at state $i$. The \emph{state} of the population at time $t$ is the number of individuals of focal type $i$ present at time $t$.

The Moran process is the stochastic process defined in such a way that a population at time $t+\Delta t$, $\Delta t>0$ fixed, is built from the population at state $i$ at time $t$ in two steps: i) with equal probability an individual is selected at random to be eliminated and ii) with probability $\tsp_i$ ($1-\tsp_i$, respect.) and individual of type \A (\B, respect.) is selected to reproduce. Since there is no mutation, it is clear that $\tsp_0=1-\tsp_N=0$.

The transition matrix of the Moran process is given by $\bM\bydef\left(M_{ij}\right)_{i,j=0,\dots,N}$, with $M_{ij}=0$ for $|i-j|>1$, $M_{i-1,i}=\frac{i}{N}(1-\tsp_i)$, $M_{i+1,i}=\frac{N-i}{N}\tsp_i$, $M_{ii}=1-M_{i+1,i}-M_{i-1,i}$. The state vector of  the population at time $t$ is given by $\bp(t)=(p_0(t),p_1(t),\dots,p_N(t))$, where $p_i(t)$ indicates the probability to find the population in state $i$ at time $t$. The evolution is given by $\bp(t+\Delta t)=\bM\bp(t)$, $\bp(0)=\bp^\ini\in\simplex^{N}\bydef\{\bx|x_i\ge 0, \sum_{i=0}^{N}x_i=1\}$ (the $N$-dimensional simplex). 

 Let $\widetilde{\bM}=\left(M_{ij}\right)_{i,j=1,\dots,N-1}$ be the \emph{core} matrix associated to the Moran process. We write $\langle\cdot,\cdot,\rangle_N$ for the Euclidean inner product in $\R^{N}$. We consider the evolution given by $\bp_{i+1}=\bM\bp_{i}$, $\bp_0=\bp^{\ini}\in\simplex^{N}$, and call $\widetilde{\bp}\in\R_+^{N-1}$ such that $\bp=p_0\oplus\widetilde{\bp}\oplus p_N$.

The following Lemma collects several  results from \cite{CMRS:21} that will be useful in the sequel:

\begin{lemma}\label{lem:PFextended}
Consider the Moran core matrix $\widetilde{\bM}$. Then, there are two bases of the space $\R^{N-1}$, $\left(\widetilde{\bu}^{(i)}\right)_{i=1,\dots,N-1}$ and $\left(\widetilde{\bv}^{(i)}\right)_{i=1,\dots,N-1}$ such that
\begin{enumerate}
\item\label{lem:PFextended:eigenvectors}  $\widetilde{\bM}\widetilde{\bv}^{(i)}=\mu_i\widetilde{\bv}^{(i)}$, $\widetilde{\bu}^{(i)}\widetilde{\bM}=\mu_i\widetilde{\bu}^{(i)}$, with $\mu_1>\mu_2\ge\dots\ge\mu_{N_0}$, where $\mu_i\in\R$ for all $i$.
\item\label{lem:PFextended:dominant} $u^{(i)}_j>0$ for all $j=1,\dots,N-1$ if and only if $i=1$. The same is true for $\widetilde{\bv}^{(i)}$. Furthermore, $\sum_{j=1}^{N-1}v^{(1)}_j=\langle\widetilde{\bu}^{(i)},\widetilde{\bv}^{(i)}\rangle=1$, $i=1,\dots, N-1$.
\item\label{lem:PFextended:extension} Each vector $\widetilde{\bu}^{(i)}$ and $\widetilde{\bv}^{(i)}$ can be extended to vectors in $\R^{N+1}$, $\bu^{(i)}\bydef u^{(i)}_0\oplus\widetilde{\bu}^{(i)}\oplus u^{(i)}_N$, with $u^{(i)}_0=u^{(i)}_N=0$ and $\bv^{(i)}\bydef v^{(i)}_0\oplus\widetilde{\bv}^{(i)}\oplus v^{(i)}_N$, respectively, such that $\bu^{(i)}\bM=\mu_i\bu^{(i)}$ and $\bM\bv^{(i)}=\mu_i\bv^{(i)}$. 
\item\label{lem:PFextended:ortogonality} $\langle\widetilde{\bu}^{(i)},\widetilde{\bv}^{(j)}\rangle=0$ for $\mu_i\ne\mu_j$ and $\sum_{j=0}^Nv^{(i)}_j=0$ for $i\ge 2$.
\item\label{lem:PFextended:reversibility} $\widetilde{\bu}^{(k)}_iM_{ij}\widetilde{\bv}^{(k)}_j=\widetilde{\bu}^{(k)}_jM_{ji}\widetilde{\bv}^{(k)}_i$.
\end{enumerate}
\end{lemma}

\begin{proof}
Items~\ref{lem:PFextended:eigenvectors} and~\ref{lem:PFextended:dominant} follow from the spectral theorem for tridiagonal matrices, the Perron-Frobenius theorem for primitive matrices and a normalization choice. Item~\ref{lem:PFextended:extension} follows from straightforward calculations. For~\ref{lem:PFextended:ortogonality}, note that $\mu_i\langle\bu^{(i)},\bv^{(j)}\rangle=\langle\widetilde{\bM}^\dagger\bu^{(i)},\bv^{(j)}\rangle=\langle\bu^{(i)},\widetilde{\bM}\bv^{(j)}\rangle=\mu_j\langle\bu^{(i)},\bv^{(j)}\rangle$, and the orthogonality relation follows; the last equation follows noting that $\bM^\dagger\bone=\bone$, where $\bone_i=1$ for $i=0,\dots,N$. Finally, note that $M_{ij}=0$ if $|i-j|>1$ and it is necessary to prove item~\ref{lem:PFextended:reversibility} only if $i=j\pm1$. For $i=1$, $j=2$, $u^{(k)}_1M_{11}v^{(k)}_1+u^{(k)}_2M_{21}v^{(k)}_1=\mu_ku^{(k)}_1v^{(k)}_1=u^{(k)}_1M_{11}v^{(k)}_1+u^{(k)}_1M_{12}v^{(k)}_2$, and the same for $i=2$, $j=1$. Using the induction principle, and 
\begin{align*}
&u^{(k)}_{j-1}M_{j-1,j}v^{(k)}_j+u^{(k)}_jM_{jj}v^{(k)}_j+u^{(k)}_{j+1}M_{j+1,j}v^{(k)}_j=\mu_ku^{(k)}_jv^{(k)}_j\\
&\quad=u^{(k)}_jM_{j,j-1}v^{(k)}_{j-1}+u^{(k)}_jM_{jj}v^{(k)}_j+u^{(k)}_jM_{j,j+1}v^{(k)}_{j+1}\ ,
\end{align*}
we finish the proof of the last item.
\end{proof}

\begin{remark}
It is not usual to define the Moran process from the type selection probability vector, but rather from the fitnesses functions $\Psi^{(\A),(\B)}:\{0,1,\dots,N\}\to\R_+$, cf.~\cite{Nowak_EvolutionaryDynamics}. Assuming the fitness functions as proxies of the probability to select type \A for reproduction when the population is at state $i$, it is customary to assume that $\tsp_i=i\Psi^{(\A)}(i)/(i\Psi^{(\A)}(i)+(N-i)\Psi^{(\B)}(i))$.
\end{remark}

\begin{remark}
 Inspired by the terminology used in the Markov Chain literature, and indeed also in~\cite{Maas_JFA}, the condition given in Lemma~\ref{lem:PFextended}.\ref{lem:PFextended:reversibility} was called \emph{microreversibility} in~\cite{CMRS:21}; however this condition is not directly related with the concept of time-reversibility used in this work. It is closely related to the concept of \emph{adiabatic}, or \emph{quasi-stationary} process, as the central idea is that each step in the Markov process is an equilibrium state. Therefore, it is no surprise that it is satisfied by general birth-and-death processes, but not by the Wright-Fisher process (to be discussed later on), when the full population is replaced in a single step.
\end{remark}

\subsection{Reversibility and irreversibility in the Moran process}
\label{ssec:irreversibility}

The evolution of a population in the Moran process is a succession of states, from a given initial condition until one of the two absorbing, final states. Here we show that certain paths are more likely than reverse  paths. Therefore, it is possible to infer, from a sequence of states, if the order corresponds to the reality or to a backward \emph{film} being presented, even if metastable interior states are present.

More precisely,
\begin{lemma}
	\label{lem:prob_paths}
Let $0<i,j<N$ and let $\overrightarrow{C}=\left(i=x_0,x_1,\dots,x_n=j\right)$ be a certain path from $i$ to $j$ (note that the path does not need to be monotonic) and $\overleftarrow{C}=\left(j=x_n,x_{n-1},\dots,x_0=i\right)$ the reverse path. The ratio between the probabilities that a stochastic process follows the path $\overrightarrow{C}$ and $\overleftarrow{C}$ is given by
\[
\razaoP=\frac{u^{(1)}_i/v^{(1)}_i}{u^{(1)}_j/v^{(1)}_j}\ .
\]
\end{lemma}

\begin{proof}
We use Lemma~\ref{lem:PFextended}.\ref{lem:PFextended:reversibility}, with $k=1$, to find
\begin{equation*}
\razaoP=\frac{\prod_{k=0}^{n-1} M_{x_{k+1},x_k}}{\prod_{k=0}^{n-1} M_{x_{k},x_{k+1}}}=\prod_{k=0}^{n-1}\frac{u^{(1)}_{x_{k}}v^{(1)}_{x_{k+1}}}{u^{(1)}_{x_{k+1}}v^{(1)}_{x_{k}}}=\frac{u^{(1)}_{x_0}v^{(1)}_{x_n}}{u^{(1)}_{x_n}v^{(1)}_{x_0}}\ ,
\end{equation*}
which finishes the proof.
\end{proof}

Note that this implies that the flow goes from the maximum of the ratio $z_i\bydef u^{(1)}_i/v^{(1)}_i$ to the boundaries, where $u^{(1)}_0=u^{(1)}_N=0$, and therefore, where $z_i\ge 0$ reaches its minimum.

In particular, in the neutral case, we have, for $0<i<N$,  that $v^{(1)}_i=1$ and $w^{(1)}_i=i(N-i)/C_N$, where $C_N$ is a normalising constant. Hence, we obtain  $\razaoP=\frac{i(N-i)}{j(N-j)}$. \footnote{This derivation can also be obtained using the eigenvector structure of the core matrix. Indeed \[\razaoP = \frac{\mathbf{e}_j\bM^n\mathbf{e}_i}{\mathbf{e}_i\bM^n\mathbf{e}_j},\] and the expressions on the RHS can be expanded  with respect to the eigenbasis of the core.}
It is clear that if $j\approx 0$ or $j\approx N$, $\razaoP\gg1$ and therefore paths linking interior points to the boundary are more likely that the reverse evolution. In fact, invasion is a rare process that occurs only because mutations are frequent (see, e.g.,~\cite{crow1970introduction}).

\begin{remark}
 While in physics, a typical macroscopic system has $10^{23}$ degrees of freedom, numbers in biology are far below, ranging from $10^6$ to $10^9$ individuals. Therefore, an \emph{irreversible} process in physics is normally linked to impossibility (recurrence time of the order of magnitude of the age of the universe), while in biology these events, although unlikely for a short (say, human scale), will be very likely in long (say, geological) times. See also the discussion on non-increasing fixation probabilities for the Wright-Fisher process in~\cite{ChalubSouza18}.
\end{remark}

\begin{remark}\label{rmk:wsp}
 In the large population, weak selection limit, it is useful to assume 
\begin{equation}\label{eq:tsp_fitness}
 \tsp_i=\frac{i}{N}\left[1-\frac{1}{\kappa N}\frac{N-i}{N}V'\left(\frac{i}{N}\right)\right]\ ,
\end{equation}
where $V:[0,1]\to\R$ is the so-called \emph{fitness potential} --- as defined in ; see also~\cite{ChalubSouza_JMB2016}. It is also worth pointing out that this includes the standard two-player games as a special case, as this is equivalent to the choice of a quadratic potential $V$.
In the continuous limit, it follows that (see \cite{CMRS:21})
\[
\lim_{N\to\infty} \frac{u^{(1)}_{\lfloor xN\rfloor}}{Nv^{(1)}_{\lfloor xN\rfloor}}\to C x(1-x)\e^{2V(x)/\kappa}\ ,\qquad x\in(0,1)\ ,\qquad C\in\R\ .
\]
Therefore, in view of Lemma~\ref{lem:prob_paths}, we expect for $N$ large that
\[
\frac{P_{x\to y}}{P_{y\to x}}=\frac{x(1-x)}{y(1-y)}\e^{\frac{2}{k}\left(V(x)-V(y)\right)}\ .
\]
where $P_{x\to y}$ is the probability for a stochastic process to go from state $x$ to $y$.
We conclude that, if $V(x)>V(y)$ ($V(x)<V(y))$, then for $k\ll 1$, $\frac{P_{x\to y}}{P_{y\to x}}=\mathcal{O}(\exp(Ck^{-1}))$ ($\mathcal{O}(\exp(-Ck^{-1}))$, respect.).  Therefore, the probability mass that remains in the interior, i.e. conditioned on non-extinction, flows  from larger values of the potential $V$ to smaller values of $V$. Hence, for small $k$,  the quasi-stationary probability should peak at minima of $V$. Notice also that as either  $x$ or $y$ approaches either 0 or 1 then this ratio tends to 0 or infinity --- this is a consequence of these states being absorbing. This also gives a preferred direction (\emph{arrow of time}) for the evolutionary process.
\end{remark}

From Lemma~\ref{lem:prob_paths} we conclude that any interior local minima of $\sfrac{u^{(1)}_i}{v^{(1)}_i}$ is such that an initial state sufficiently close to it will be initially be attracted to this point. Namely, the outflow of probability mass on site $i$ is smaller than the corresponding inflow, resulting in an equilibrium with local concentration of probability. As a consequence of Lemma~\ref{lem:PFextended}, global minima of $\sfrac{u^{(1)}_i}{v^{(1)}_i}$ are always on the boundaries, the attractors of the Moran process.
The relation between these local-in-time, local-in-space attractors (in the loose definition we sketched, but did not formalize above) and metastable states will be further explored in a forthcoming work; see, however,~\cite{ChalubSouza18}.

\section{The entropy}
\label{sec:entropy}

Let $\bp\in\simplex^{N}\bydef\{\bx\in\R_+^{N+1}|\sum_ix_i=1\}$. We define the entropy of a population at state $\bp$ evolving by the Moran process given by $\bM$ by
\begin{equation}\label{def:entropy}
E(\mathbf{p})=\sum_{i=1}^{N-1}\phi\left(\frac{p_i}{v^{(1)}_i\langle\bu^{(1)},\mathbf{p}\rangle_{N+1}}\right)v^{(1)}_iu^{(1)}_i\ ,
\end{equation}
where $\phi$ is a concave function such that
 $\phi(0)=\phi(1)=0$. Note that, as the summation is from 1 to $N-1$, we may use indistinctly $\widetilde{\bv}^{(i)}$, $\widetilde{\bu}^{(i)}$ or $\bv^{(i)}$, $\bu^{(i)}$, respectively. Furthermore $\langle\bu^{(1)},\bp\rangle_{N+1}=\langle\widetilde{\bu}^{(i)},\widetilde{\bp}\rangle_{N-1}$.

Note that
$E(\bp)$ depends on both the species' features (i.e., on the type selection vector or on the fitness) and on the population state $\bp$ at any given time. Therefore, $E$ is a state-dependent function.

 \begin{remark}\label{rmk:BGS-Tsallis}
As it is customary in mathematics, we opt for a convex function, and therefore the sign is reversed when compared to most texts in physics. This change is immaterial, except that the second law of thermodynamics implies that entropy \emph{decreases} in time. As particular examples, if $\phi(x)=\phi_1(x)\bydef x\log x$, we say that the associated entropy is the Boltzmann-Gibbs-Shannon entropy (or BGS) $E_{\mathrm{BGS}}=E$. If $\phi(x)=\phi_m(x)\bydef\frac{x^m-x}{m-1}$, $m\ne 1$, we call $E_m=E$ the Tsallis $m$-entropy. It is clear that $\lim_{m\to 1}E_m=E_{\mathrm{BGS}}$.
 \end{remark}

 \begin{remark}
  
 Given a vector $\bF=(F_0,F_1,\dots,F_N)$, with $0=F_0<F_1<F_2\dots<F_N=1$, there exists a unique Moran process $\bM$ such that $\bF$ is its fixation vector \cite{ChalubSouza:2017a}. Therefore, from $\bF$, we obtain unique $\bu^{(i)}$, and $\bv^{(i)}$ (in particular, for $i=1$), and therefore, the entropy is uniquely determined by the final evolution of the Moran process. This show that  time (in this case, the number of steps $n$) does not play a role in the entropy formula, even indirectly, showing that the entropy depends only on the present state (given the two types in the population) and not in the previous evolutionary history --  a function of the \emph{state}, not of the \emph{path}. 
 
 On the other hand, given $\bu^{(1)}$ and $\bv^{(1)}$  the entropy $E$ is well defined, but not the Moran process. In fact noting that from Lemma~\ref{lem:PFextended}.\ref{lem:PFextended:reversibility}
 \[
 u^{(1)}_i\frac{i+1}{N}(1-\tsp_{i+1})v^{(1)}_{i+1}=u^{(1)}_{i+1}\frac{N-i}{N}\tsp_iv^{(1)}_i\ ,
 \]
we conclude that the type selection probability can be determined only after imposing the value of $\tsp_i$ for a certain value $i$.
 \end{remark}
 
 An important final comment on expression~\eqref{def:entropy} is that it is optimal in a very precise sense. In fact, if we measure the distance between two probability distributions using the Wasserstein-Shashahani distance (see Remark below), then the evolution given by the Moran process gives the path of maximum entropy decrease along all possible evolutionary trajectories that satisfy natural conservation laws of the Moran process --- cf. \cite{CMRS:21}.

\begin{remark}
 The Shashahani distance between $x,y\in[0,1]$ was explicitly introduced in~\cite{Shahshahani_1979}, but appeared previously in~\cite{cavalli1999genetics}. In the one-dimensional case, it is given by $\int_x^y\frac{\rd z}{\sqrt{z(1-z)}}=2\left|\arcsin\sqrt{x}-\arcsin\sqrt{y}\right|$. The Wasserstein-Shashahani distance is the transport distance between probability measures in a space where point distances are given by the Shashahani distance.
\end{remark}

\subsection{Entropy decay}
\label{ssec:decay}

\begin{lemma}\label{lem:decay}
For any $\bp\in\simplex^{N-1}$, it is true that
$E(\bM\bp)\le E(\bp)$. If $\phi$ is strictly convex, $E(\bM\bp)=E(\bp)$ if and only if $\bp=\alpha\be_0+\lambda_1\bv^{(1)}+\beta\be_N$, for $\alpha,\beta,\lambda_1\ge 0$, $\alpha+\beta+\lambda_1=1$.
\end{lemma}

\begin{proof}
\begin{align*}
E(\bM\bp)&=\sum_{i=0}^{N-1}\phi\left(\sum_{j=0}^N\frac{M_{ij}p_j}{v^{(1)}_i\langle\bu^{(1)},\bM\bp\rangle}\right)u^{(1)}_iv^{(1)}_i\\
&=\sum_{i=0}^N\phi\left(\sum_{j=0}^N\frac{u^{(1)}_iM_{ij}v^{(1)}_jp_j}{u^{(1)}_iv^{(1)}_jv^{(1)}_i\langle\bM^\dagger\bu^{(1)},\bp\rangle}\right)u^{(1)}_iv^{(1)}_i\\
&=\sum_{i=0}^N\phi\left(\sum_{j=0}^N\frac{u^{(1)}_jM_{ji}v^{(1)}_ip_j}{\mu_1u^{(1)}_iv^{(1)}_jv^{(1)}_i\langle\bu^{(1)},\bp\rangle}\right)u^{(1)}_iv^{(1)}_i\\
&\le\sum_{i,j=0}^N\frac{u^{(1)}_jM_{ji}}{\mu_1u^{(1)}_i}\phi\left(\frac{p_j}{v^{(1)}_j\langle\bu^{(1)},\bp\rangle}\right)u^{(1)}_iv^{(1)}_i\\
&=\sum_{j=1}^{N-1}\phi\left(\frac{p_j}{v^{(1)}_j\langle\bu^{(1)},\bp\rangle}\right)u^{(1)}_jv^{(1)}_j=E(\bp)\ .
\end{align*}
First and last equalities follow from Definition~\ref{def:entropy}; second equality is a simple manipulation and it follows from properties of adjoint matrices; third equality is a consequence of Lemma~\ref{lem:PFextended}, properties~\ref{lem:PFextended:eigenvectors} and~\ref{lem:PFextended:reversibility}. The central step in the above derivation is the inequality at the fourth line, in which convexity of $\phi$ is explicitly used, together with the fact that $\sum_{j=0}^N\frac{u_j^{(1)}M_{ji}}{\mu_1u_i^{(1)}}=1$. Finally, we use that $\sum_{i=0}^NM_{ji}v_i^{(1)}=\mu_1v_j^{(1)}$ and prove the inequality.

If we assume that $\phi$ is strictly convex, the inequality in the fourth line will be strict unless $\frac{p_j}{v^{(1)}_j\langle\bu^{(1)},\bp\rangle}$ is independent of $j$, for $j=1,\dots,N-1$. Finally, if $\bp=\alpha\be_0+\lambda_1\bv^{(1)}+\beta\be_N$, then $E(\bM\bp)=E(\bp)=\phi(1)=0$.
\end{proof}

The last result shows that the entropy changes whenever there is a change in $\bp$ with respect to the quasi-stationary measure $\alpha\be_0+\lambda_1\bv^{(1)}+\beta\be_N$. In this sense, the entropy is a measure of the mixing of a population evolving according to the process $\bM$ \emph{with respect} to the generalised quasi-stationary measure of $\bM$.

\subsection{Asymptotic behaviour}
\label{ssec:asymptotic}

\begin{lemma}\label{lem:asymptotic_decay}
Assume $\mu_i\ne\mu_j$ for all $i\ne j$ and let $\bp$ be such that $\bp\not\in\mathrm{span}\{\be_0,\bv^{(1)},\be_N\}$. Then, there is $j\in\{2,\dots,N-1\}$ such that
\begin{equation}\label{eq:asymptotic_decay}
 E(\bM^n\bp)\approx\e^{-2n\log\frac{\mu_1}{\mu_j}}\frac{m\langle\bu^{(j)},\bp\rangle^2}{2\langle\bu^{(1)},\bp\rangle^2}\sum_{i=1}^{N-1}\left(v^{(j)}_i\right)^2\frac{u^{(1)}_i}{v^{(1)}_i} ,
 \end{equation}
 when $n\to\infty$.
\end{lemma}

\begin{proof}

Let $\bp=\alpha\be_0+\sum_{j=1}^{N-1}\lambda_j\bv^{(j)}+\beta\be_N$ and let $j_*=\min\{i\ge 2|\lambda_i\ne0\}$. 
It is clear that $j_*$ is well defined, otherwise $\bp\in\mathrm{span}\{\be_0,\bv^{(1)},\be_N\}$. Furthermore, $\lambda_1=\langle\bu^{(1)},\bp\rangle>0$.

On one hand, 
$\langle \bu^{(1)},\bM^n\bp\rangle=\langle\left(\bM^\dagger\right)^n\bu^{(1)},\bp\rangle=\mu_1^n\langle\bu^{(1)},\bp\rangle=\mu_1^n\lambda_1$.
 On the other hand, for $i=1,\dots,N-1$,
 \begin{align*}
 \left(\bM^n\bp\right)_i&=\left\langle\sum_{j=1}^{N-1}\lambda_j\mu_j^n\bv^{(j)},\mathbf{e}_i\right\rangle=\sum_{j=1}^{N-1}\lambda_j\mu_j^n\left\langle\bv^{(j)},\mathbf{e}_i\right\rangle\\
 &=\lambda_1\mu^n_1v^{(1)}_i+\sum_{j=j_*}^{N-1}\mu_j^n\lambda_jv^{(j)}_i
 \end{align*}
 
 Therefore,
 \begin{align*}
 E(\bM^n\bp)&=\sum_{i=1}^{N-1}\phi\left(\frac{\left(\bM^n\bp\right)_i}{v^{(1)}_i\langle\bu^{(1)},\bM^n\mathbf{p}\rangle}\right)v^{(1)}_iu^{(1)}_i\\
 &=\sum_{i=1}^{N-1}\phi\left(1+\sum_{j=j_*}^{N-1}\left(\frac{\mu_j}{\mu_1}\right)^{n}\frac{\lambda_jv^{(j)}_i}{\lambda_1 v^{(1)}_i}\right)v^{(1)}_iu^{(1)}_i\\
 &\approx\sum_{i=1}^{N-1}\biggl\{\phi(1)+\phi'(1)\sum_{j=j_*}^{N-1}\left(\frac{\mu_j}{\mu_1}\right)^{n}\frac{\lambda_jv^{(j)}_i}{\lambda_1 v^{(1)}_i}
 +\frac{\phi''(1)}{2}\left(\frac{\mu_{j_*}^n\lambda_{j_*}v^{(j_*)}_i}{\mu_1^n\lambda_1 v^{(1)}_i}\right)^2\biggr\}v^{(1)}_iu^{(1)}_i
 \end{align*}

 For the BGS entropy and Tsallis entropy, $\phi(1)=0$, $\phi'(1)=1$ and $\phi''(1)=m$ ($m=1$ in the BGS case and $m\ne 1$ in the Tsallis case).
 
 Finally
 \begin{align*}
 E(\bM^n\bp)&\approx\sum_{j=j_*}^{N-1}\left(\frac{\mu_j}{\mu_1}\right)^n\frac{\lambda_j}{\lambda_1}\underbrace{\sum_{i=1}^{N-1}v^{(j)}_iu^{(1)}_i}_{=\langle\bv^{(j)},\bu^{(1)}\rangle=0}+\frac{m}{2}\left(\frac{\mu_{j_*}}{\mu_1}\right)^{2n}\frac{\lambda_{j_*}^2}{\lambda_1^2}\sum_{i=1}^{N-1}\left(v^{(j_*)}_i\right)^2\frac{u^{(1)}_i}{v^{(1)}_i}\\
 &=\e^{-2n\log\frac{\mu_1}{\mu_{j_*}}}\frac{m\langle\bu^{(j_*)},\bp\rangle^2}{2\langle\bu^{(1)},\bp\rangle^2}\sum_{i=1}^{N-1}\left(v^{(j_*)}_i\right)^2\frac{u^{(1)}_i}{v^{(1)}_i}\ .
 \end{align*}
 
 \end{proof}
 
 
 \begin{cor}\label{cor:quasi_stationary}
Let $\bp\in\simplex^{N+1}$. Then $\bp\not\in\mathrm{span}\{\be_0,\bv^{(1)},\be_N\}$ if and only if there exists $n_0\in\mathbb{N}$ and $\bar p\in\simplex^{N+1}$ such that for $n>n_0$, $E(\bM^n\bar\bp)<E(\bM^n\bp)$.
 \end{cor}
 
 \begin{proof}
 $\Rightarrow$  Let $j$ be as in Lemma~\ref{lem:asymptotic_decay}.
 Note that $\sum_{i=1}^{N-1}\left(v^{(j)}_i\right)^2\frac{u^{(1)}_i}{v^{(1)}_i}>0$ does not depend on the initial condition and consider $\bar\bp\in\simplex^{N+1}$ such that $|\langle\bu^{(j)},\bar\bp\rangle|<|\langle\bu^{(j)},\bp\rangle|$. $\Leftarrow$. Assume $\bp\in\mathrm{span}\{\be_0,\bv^{(1)},\be_N\}$; then $E(\bM^n\bp)=0$ for all $n$, and for any $\bp\in\simplex^{N+1}$, $E(\bp)\ge 0$.
 \end{proof}
 
In the generic case (i.e., if $\lambda_2\ne 0$), and assuming the weak selection principle (cf. Remark~\ref{rmk:wsp}), the entropy decay rate is proportional to the spectral gap of $\bM$. Minimum entropy occur for $\bp=\alpha\be_0+\lambda_1\bv^{(1)}+\beta\be_N$.

\begin{cor}\label{cor:entropy_decay_WSP}
 Assume the same conditions as Lemma~\ref{lem:decay}, assume that the weak selection principle, Equation~\eqref{eq:tsp_fitness}, is satisfied for a certain smooth function $V$ and consider $\lambda_2\ne 0$. If $N\gg 1$, the decay rate of the entropy is proportional to the spectral gap of $\widetilde{\bM}$.
\end{cor}

\begin{proof}
From Lemma~\ref{lem:PFextended}, we have that $\mu_1=\frac{\sum_{i,j=1}^{N-1}M_{ij}v_i^{(1)}}{\sum_{j=1}^{N-1}v^{(1)}_j}$, i.e., $\mu_1$ is the weighted average, with weights given by $v_i^{(1)}>0$, of the row sums $\sum_{j=1}^{N-1}M_{ij}$. Therefore
 \[
  \min_{i\in\{1,\dots,N-1\}}\sum_{j=1}^{N-1}M_{ij}<\mu_1<\max_{i\in\{1,\dots,N-1\}}\sum_{j=1}^{N-1}M_{ij}\ ,
 \]
and, finally.
\[
 \min\left\{1-\frac{1}{N}(1-s_1),1-\frac{s_{N-1}}{N-1}\right\}<\mu_1<1\ ,
\]
and, therefore $\lim_{N\to\infty}\mu_1=1$. 

With the assumption of weak selection, 
\[
M_{i\pm1,i}=\frac{i(N-i)}{N^2}\left[1\mp\frac{i(N-i)}{\kappa N^3}V'\left(\frac{i}{N}\right)\right]\ ,
\]
i.e., $\bM=\bM^{(\mathrm{N})}\left[\mathrm{I}+\frac{1}{\kappa N}\mathbf{D}\right]$, where $\mathbf{M}^{(\mathrm{N})}$ is the neutral Moran matrix, and $\mathbf{D}$ is a tridiagonal perturbation matrix. From the functional continuity of the spectrum with respect to the matrices entries --- see \cite[Theorem 5.2]{Kato_1995} and from the fact that the eigenvalues of the neutral Moran process are given by $\mu_i^{\mathrm{(N)}}=1-\frac{i(i+1)}{N^2}$~\cite{Moran}, we conclude that $\mu_i\lesssim 1$ for $N\gg 1$.
Finally $\log\frac{\mu_1}{\mu_{2}}\approx \mu_1-\mu_{2}$, and this finishes the proof.
\end{proof}

\begin{remark}

The rate at which $\bM^n$ approaches the stationary distribution in $\mathop{\mathrm{span}}\{\be_0,\be_N\}$ is given by  $\mu_1$, the leading eigenvalue of $\widetilde{\bM}$. On the other hand, assuming $\lambda_2\ne0$, and $\bp\not\in\mathop{\mathrm{span}}\{\be_0,\bv^{(1)},\be_N\}$ the entropy decays according to the spectral gap of the core matrix $\widetilde{\bM}$, $\mu_1-\mu_2$.
The latter value measures the rate in which the interior points of the vector $\bM^n\bp$ approaches the quasi-stationary distribution, while the former measures the rate in which the probability distribution approaches the final stationary distribution. --- cf. \cite{meleard2012quasi}. See \cite{ChalubSouza14b} for a similar calculation for the stochastic SIS process and \cite{velleret2019two,velleret2020individual} for results on nested Moran processes.
\end{remark}


 
\subsection{Additivity, subadditivity and multiloci evolution}
\label{ssec:multiloci}

In this subsection (and only here), we use $\bz=\bv^{(1)}$ and $\bw=\bu^{(1)}$.

 Consider $n\ge 2$ independently evolving loci in a population of fixed size. Each loci is occupied by one out of two alleles.
 Therefore, let $p_{i_1i_2\dots i_n}$ be the probability to find the population at state $i_j$, when only the locus $j$ is taken in consideration.

 Therefore, the full system can be described using a tensor product representation. Indeed, let us assume that, for each locus, evolution is governed by a two types Moran process. We write $\bM_k$ and $\bp^{(k)}$ for the corresponding Moran transition matrix and probability vector, respectively. In addition, let  $\bp\bydef\bigotimes_{k=1}^n \bp^{(k)}$ and $\bM \bydef \bigotimes_{k=1}^n\bM_k$, where $\otimes$ denotes the Kronecker product of matrices --- it is easily checked that $\bM$ is column stochastic, since so are all the $\bM_k$. Then, the full dynamics is given by
 \[
 \bM\bp=\bigotimes_{k=1}^n\bM_k\bp^{(k)}.
 \]

The Matrix $\bM$ is outside the scope of~\cite{CMRS:21}, but for every matrix $\bM_k$, the theory developed in \cite{CMRS:21} applies immediately. 
  
  We will show in the remaining of this section that the entropy defined in equation~\eqref{def:entropy} 
  can be easily extended to this case. Furthermore, assuming independent loci, BGS entropy corresponds to additivity (i.e., the entropy of the full multiloci system is the sum of the entropy of each locus) and Tsallis entropy implies the existence of epistasis in the system. Note that both BGS and Tsallis entropy discussed here are, in fact, generalizations for processes with absorbing states of  standard definitions for irreducible processes.
 
 The discussion below provides a consistent definition of entropy as informational entropy (i.e., in the case of epistasis, the knowledge of the macroscopic observable and one microscopic variable provides information on the other variables; in the non-epistatic case, this is not true).

We assume that in the multiloci case, the entropy of independently evolving alleles is given by
\[
E(\bp)=\sum_{\bi}\phi\left(\frac{p_\bi}{z_\bi\langle\bw,\bp\rangle}\right)z_\bi w_\bi\ .
\]

From the definition of $\bM$, it follows that $\bz=\otimes_k\bz^{(k)}$, $\bw=\otimes_k\bw^{(k)}$. Finally $\langle\bw,\bp\rangle=\prod_k\langle\bw^{(k)},\bp^{(k)}\rangle$.

Firstly, let us show that the BGS entropy is additive, i.e., if we assume that  $\phi(x)=x\log x$, then
\[
E(\bp)=\sum_k E(\bp^{(k)})\ .
\]

Indeed,
\begin{align*}
E(\bp)&=\sum_\bi\frac{p_\bi}{z_i\langle\bw,\bp\rangle}\log\frac{p_\bi}{z_i\langle\bw,\bp\rangle}z_\bi w_\bi\\
&=\sum_\bi\frac{p_\bi w_\bi}{\langle\bw,\bp\rangle}\log\prod_k\frac{p^{(k)}_{i_k}}{z^{(k)}_{i_k}\langle\bw^{(k)},\bp^{(k)}\rangle}\\
&=\sum_{\bi}\sum_k\frac{\bp_\bi\bw_\bi}{\langle\bw,\bp\rangle}\log\frac{p^{(k)}_{i_k}}{z^{(k)}_{i_k}\langle\bw^{(k)},\bp^{(k)}\rangle}\\
&=\sum_k\sum_{i_k}\sum_{i_j,\,j\ne k}\frac{\bp_\bi\bw_\bi}{\langle\bw,\bp\rangle}\log\frac{p^{(k)}_{i_k}}{z^{(k)}_{i_k}\langle\bw^{(k)},\bp^{(k)}\rangle}\\
&=\sum_k\sum_{i_k}\frac{p^{(k)}_{i_k}w^{(k)}_{i_k}}{\langle\bw^{(k)},\bp^{(k)}\rangle}\log\frac{p^{(k)}_{i_k}}{z^{(k)}_{i_k}\langle\bw^{(k)},\bp^{(k)}\rangle}=\sum_kE(\bp^{(k)})\ .
\end{align*}

This also suggests that additivity of the entropy should be equivalent to assume BGS entropy --- this turns out to be correct, but we will refrain to discuss this further. In any case, if two loci have independent effects (i.e., knowledge of the frequency of one particular allele does not provide information of any kind to the current status of the other generation), then it is natural to assume additive entropy and therefore, BGS entropy.

If both alleles are subject to similar (correlated) selective forces, such that the result of one evolution provides information on the other allele evolution, the level of information provided by both evolutions will be smaller than the sum of information provided by each case. This is the case in which we shall use Tsallis $m$-entropy (see Remark.~\ref{rmk:BGS-Tsallis}).

We start by $(xy)^m-xy=(x^m-x)y+x(y^m-y)+(x^m-x)(y^m-y)$, and therefore
\begin{equation*}
E(\bp)=\frac{1}{m-1}\sum_{i=1}^{N-1}\left[\left(\frac{ p_i}{z_i\langle\bw,\bp\rangle}\right)^m-\frac{p_i}{z_i\langle\bw,\bp\rangle}\right]w_iz_i\ .
\end{equation*}

Assume (as before) the evolution of two independent alleles:
{\footnotesize
\begin{align*}
E(\bp)&=\frac{1}{m-1}\sum_{i,j=1}^{N-1}\left[\left(\frac{p^{(1)}_ip^{(2)}_j}{z^{(1)}_iz^{(2)}_j\langle\bw^{(1)},\bp^{(1)}\rangle\langle\bw^{(2)},\bp^{(2)}\rangle}\right)^m-\frac{p^{(1)}_ip^{(2)}_j}{z^{(1)}_iz^{(2)}_j\langle\bw^{(1)},\bp^{(1)}\rangle\langle\bw^{(2)},\bp^{(2)}\rangle}\right]w^{(1)}_iw^{(2)}_jz^{(1)}_iz^{(2)}_j\\
&=\frac{1}{m-1}\sum_{i,j=1}^{N-1}\Biggl\{\left[\left(\frac{p^{(1)}_i}{z^{(1)}_i\langle\bw^{(1)},\bp^{(1)}\rangle}\right)^m-\frac{p^{(1)}_i}{z^{(1)}_i\langle\bw^{(1)},\bp^{(1)}\rangle}\right]\frac{p^{(2)}_j}{z^{(2)}_j\langle\bw^{(2)},\bp^{(2)}\rangle}\\
&\quad+\frac{p^{(1)}_i}{z^{(1)}_i\langle\bw^{(1)},\bp^{(1)}\rangle}\left[\left(\frac{p^{(2)}_j}{z^{(2)}_j\langle\bw^{(2)},\bp^{(2)}\rangle}\right)^m-\frac{p^{(2)}_j}{z^{(2)}_j\langle\bw^{(2)},\bp^{(2)}\rangle}\right]\\
&\quad+\left[\left(\frac{p^{(1)}_i}{z^{(1)}_i\langle\bw^{(1)},\bp^{(1)}\rangle}\right)^m-\frac{p^{(1)}_i}{z^{(1)}_i\langle\bw^{(1)},\bp^{(1)}\rangle}\right]\left[\left(\frac{p^{(2)}_j}{z^{(2)}_j\langle\bw^{(2)},\bp^{(2)}\rangle}\right)^m-\frac{p^{(2)}_j}{z^{(2)}_j\langle\bw^{(2)},\bp^{(2)}\rangle}\right]\Biggr\}z_i^{(1)}w_i^{(1)}z_j^{(2)}w_j^{(2)}\\
&=E(\bp^{(1)})+E(\bp^{(2)})+(m-1)E(\bp^{(1)})E(\bp^{(2)})\ .
\end{align*}
}

Note that $E(\bp)\ge 0$ for all $\bp$. 

If $m>1$, then $E(\bp)\ge E(\bp^{(1)})+E(\bp^{(2)})$. Considering that as times pass (i.e., as we gain information on the possible outcomes of the evolution), the entropy decreases, this means that whenever we have \emph{correlated} evolution, the information provided by the joint evolution will be smaller than the information provided by the separated loci.
 
\begin{cor}
For any $\bp\in\simplex^{N}$, $E(\bp)\ge 0$. Furthermore, if $\phi$ is strictly convex (i.e., $\phi(x)=0$ if and only if $x=0,1$), then $E(\bp)= 0$ if, and only if, $\bp=\alpha\be_0+\lambda_1\bv^{(1)}+\beta\be_N$.
\end{cor}

\begin{proof}
The non-negativity of $E(\bp)$ follows from the non-negativity of $\phi$. For the last assertion, $E(\bp)=0$ implies that $\frac{p_i}{v^{(1)}_i\langle\bu^{(1)},\bp\rangle}\in\{0,1\}$ for all $i$. Assume, there is $i_1,i_2\in\{1,\dots,N-1\}$ such that $p_{i_1}=0$ and $p_{i_2}=v^{(1)}_{i_2}\langle\bu^{(1)},\bp\rangle\ne 0$. Multiplying $p_i$ by $\bu^{(1)}_i$ and adding over $i$, we conclude that $\langle\bu^{(1)},\bp\rangle<\langle\bu^{(1)},\bp\rangle$, which provides a contradiction, and this finishes the proof.
\end{proof}

\section{Beyond the original definition}
\label{sec:beyond}

In this section, we explore some of the properties of Entropy given by Equation~\eqref{def:entropy} going beyond the results we were able to provide rigorous proofs.

\subsection{The Wright-Fisher process}
\label{ssec:WF}

The Wright-Fisher process was introduced in~\cite{wright1931evolution} (see also \cite{crow1970introduction}) and it is a Markov process defined by a transition matrix $M_{ij}=\binom{N}{i}\tsp_j^i(1-\tsp_j)^{N-i}$, where $\tsp_i$ is the probability to sample for reproduction the focal type, when there are $i$ individuals of the focal type in a population of fixed size $N$.

The Wright-Fisher is a conundrum in the current discussion: on the one hand, as a stochastic process, it is very similar to the Moran process, even sharing the same diffusion limit. On the other hand,  we have not been able to endow the Wright-Fisher process with a variational structure, since it does not possesses the critical property stated in Lemma~\ref{lem:PFextended:reversibility}.\ref{lem:PFextended} --- see the discussion in \cite{CMRS:21}.

However, it is possible to show that the BGS entropy defined from Equation~(\ref{def:entropy}) decreases for any initial condition in the discrete time Wright-Fisher process $\bp_{n+1}=\bM\bp_n$. To prove this claim we we consider the  $Q$-process associated to the Wright-Fisher process. This is a process defined from the core matrix of the Wright-Fisher transition matrix; namely, consider the  $(N-1)\times(N-1)$ row-stochastic matrix
\[
Q_{ij}=\frac{1}{\mu_1u^{(1)}_i}M_{ij}u^{(1)}_j\ ,\quad i,j=1,\dots,N-1\ .
 \]
 The Wright-Fisher process is equivalent to the left-multiplication of the new variable $\bq=\frac{w_ip_i}{\sum_{i=1}^{N-1}w_ip_i}$.
 The stationary distribution of matrix $\bQ$ is given by $u^{(1)}v^{(1)}$. See~\cite[Section~2.3]{CMRS:21} for further details. In this context, the entropy given by Equation~\eqref{def:entropy} can be recast as the relative entropy with respect to the stationary distribution of this $Q$ process, and we can then apply a result on the decay of BGS relative entropy for irreducible processes --- cf. \cite[Theorem 4]{cover1994processes} --- to conclude that the entropy definied by Equation~\eqref{def:entropy} with the choice of BGS relative entropy is non-increasing, with equality if, and only if, $\bq$ is the stationary distribution of this $Q$ process.
 
The proof in \cite{cover1994processes} does use specific properties of the BGS entropy, and thus it cannot be readily extended to other entropies. In particular, we are not aware of any similar result to Tsallis entropies. On the other hand, the proof for the BGS entropy works for any process in the Kimura class as defined in \cite{ChalubSouza:2017a}.
{To the best of our knowledge, there is no similar result for Tsallis entropies.}

 In Figure~\ref{fig:wf} we consider several different Wright-Fisher processes, assuming weak selection, where the interaction iså given by a two-player game. Namely, we consider the interaction given by a pay-off matrix
 \[
  \mathbf{A}=\left(\begin{matrix}a&b\\ c&d\end{matrix}\right)
 \]
 Fitnesses of types \A and \B (corresponding to the first and second strategies in the pay-off matrix) are given by $\Psi^{(\A)}(x)=ax+(1-x)b$ and $\Psi^{(\B)}(x)=cx+(1-x)d$, where $x$ is the fraction of type \A individuals present in the population; finally
 \[
  \tsp_i=\frac{i\Psi^{(\A)}\left(\frac{i}{N}\right)}{i\Psi^{(\A)}\left(\frac{i}{N}\right)+(N-i)\Psi^{(\B)}\left(\frac{i}{N}\right)}\ .
 \]
In the weak selection assumption, matrix entries $a,b,c,d$ are such that $(a-1)N$ is finite when $N\to\infty$ and similarly for the other entries.

Figure~\ref{fig:wf} indeed confirms this decaying behaviour and brings some new information on the behaviour of the entropy for different initial conditions --- see the discussion in the next subsection.



\begin{figure}
\includegraphics[width=0.45\linewidth]{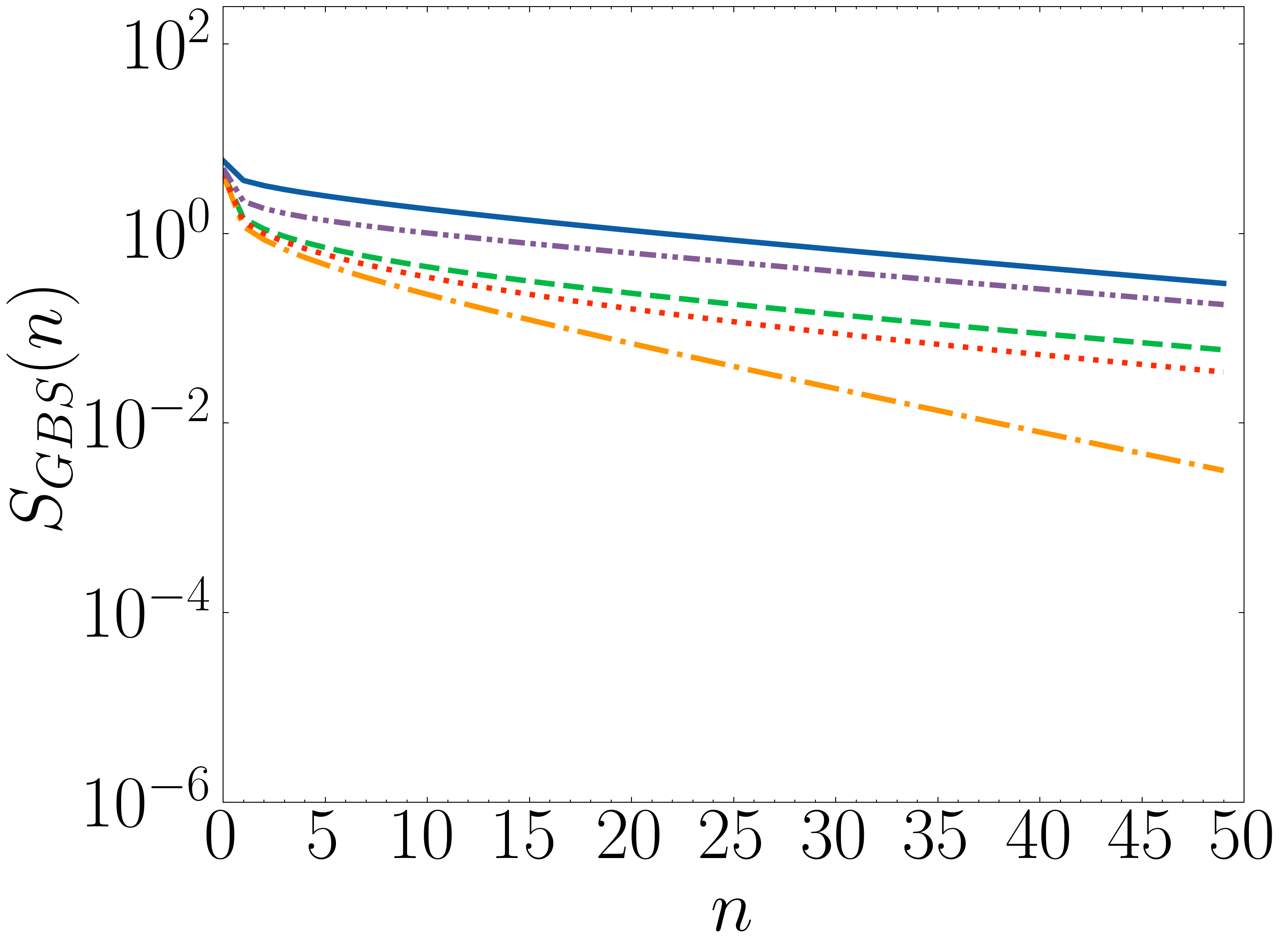}
\includegraphics[width=0.45\linewidth]{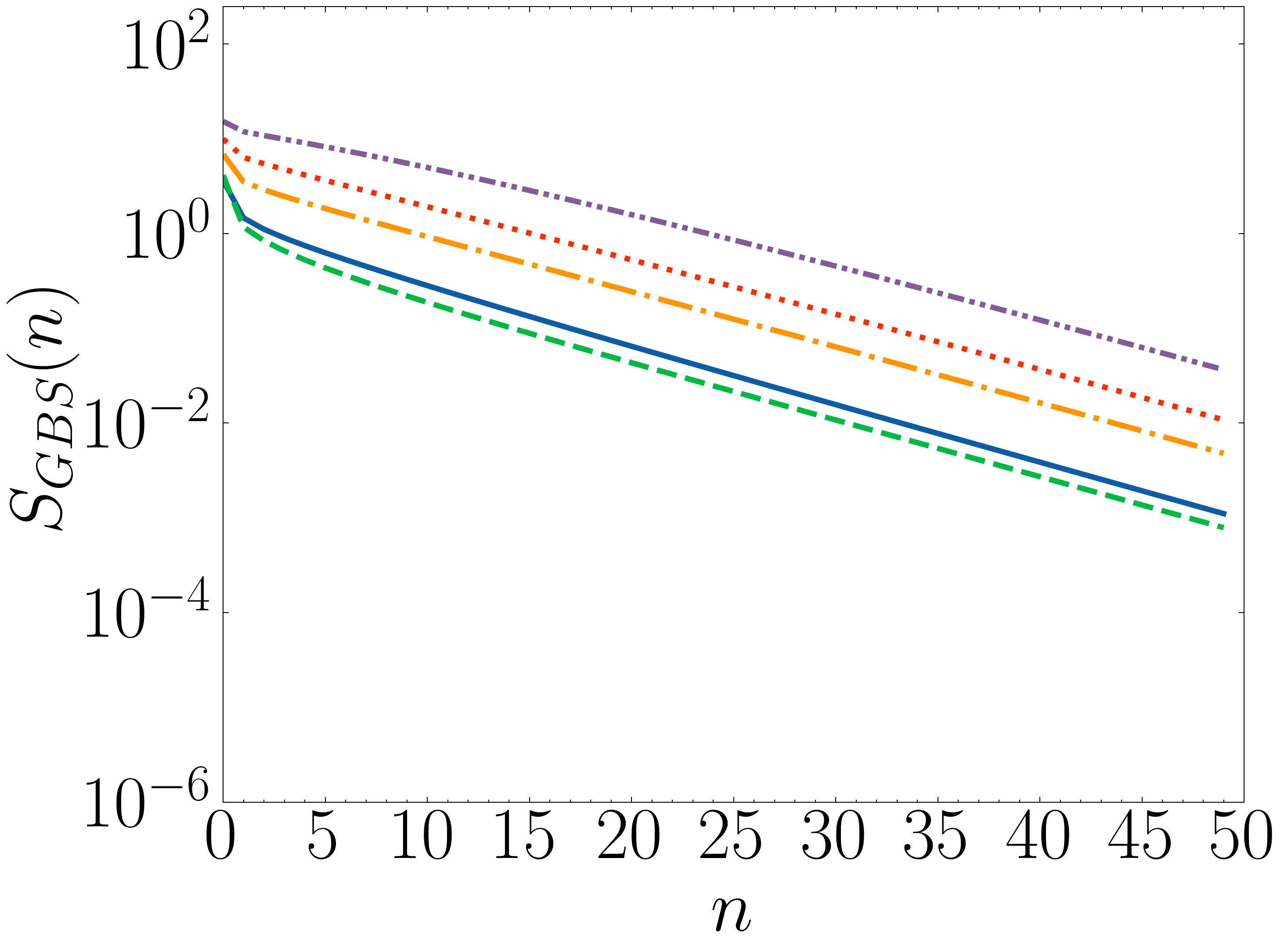}\\
\includegraphics[width=0.45\linewidth]{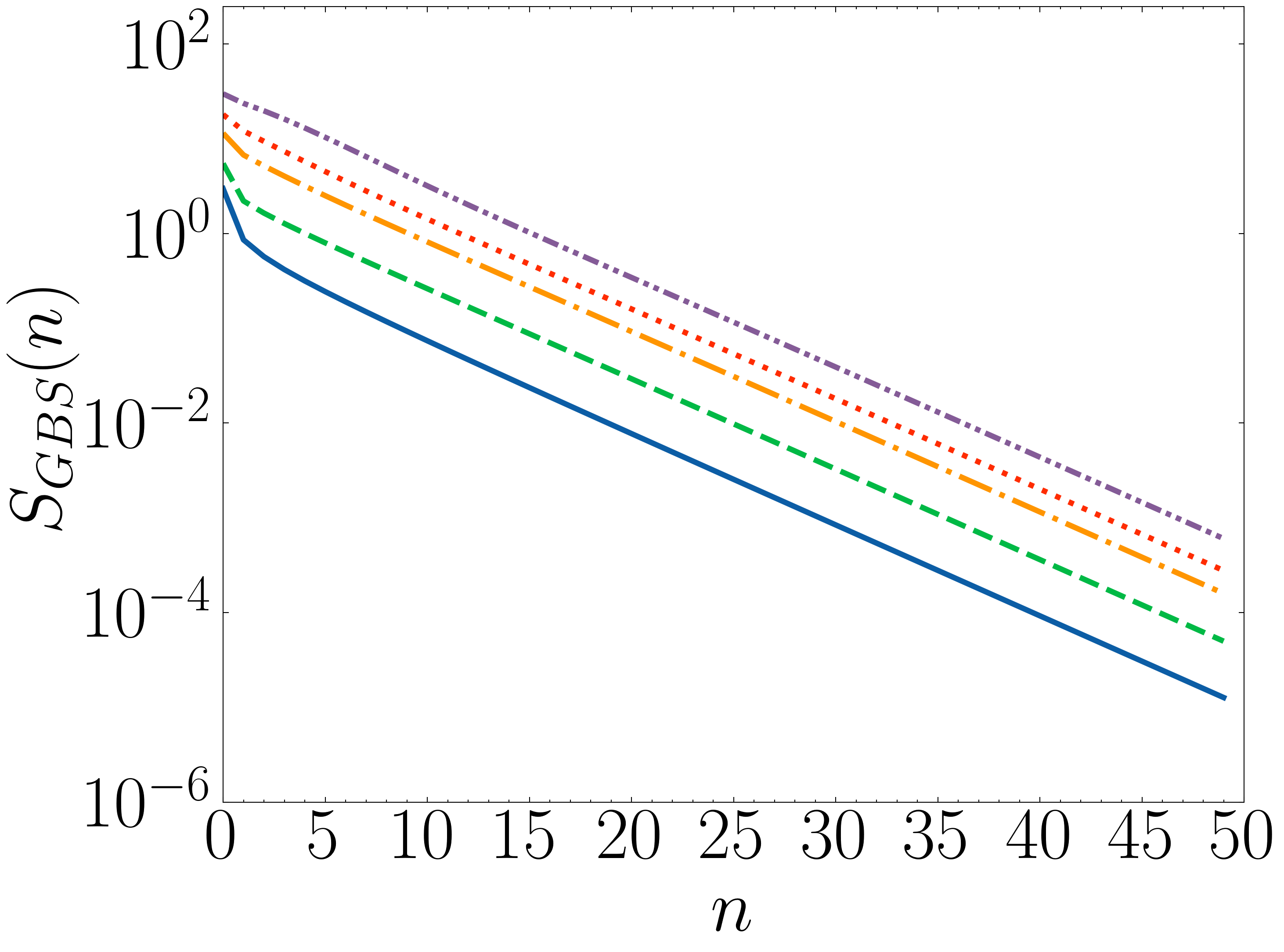}
\includegraphics[width=0.45\linewidth]{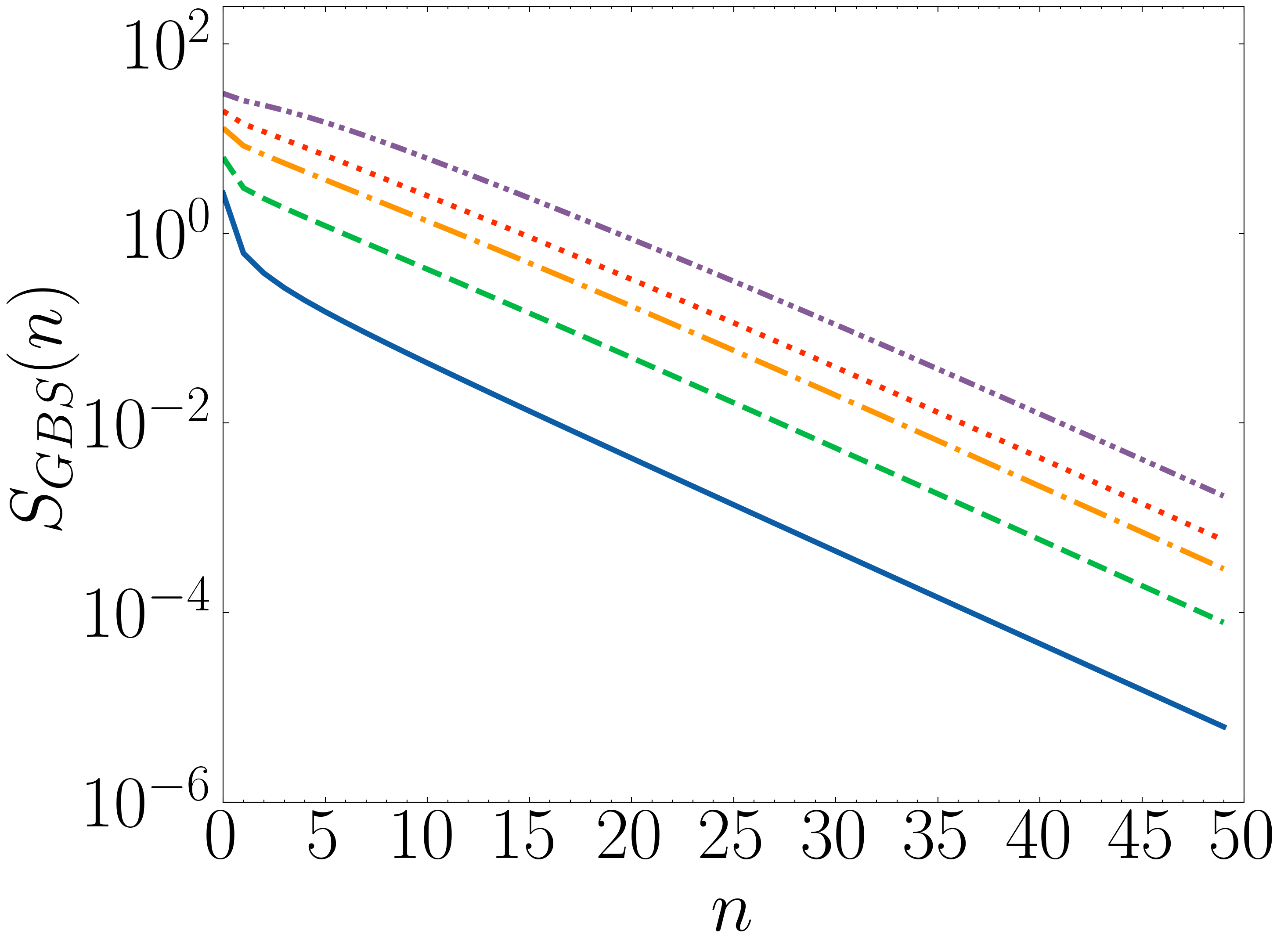}
\caption[x]{Entropy decay for the Wright-Fisher process with population $N=100$. Blue (solid), Green (dashed), Yellow (dash-dotted), Red (dotted) and Purple (dash-double-dotted) represent the initial fraction of type $\mathbb{A}$ individuals with values 0.05, 0.3, 0.5, 0.65 and 0.85, respectively. Type selection probabilities are given by formula~\eqref{eq:tsp_fitness}, with $\Psi^{(\A,\B)}$ given by game-theory: Upper left: Neutral case, with payoff matrix given by $\left(\begin{smallmatrix}1&1\\1&1\end{smallmatrix}\right)$; upper right: Dominance, $\left(\begin{smallmatrix}1&1\\1&1\end{smallmatrix}\right)+\frac{1}{N}\left(\begin{smallmatrix}0.1&0.3\\0&0\end{smallmatrix}\right)$; lower left: Coexistence $\left(\begin{smallmatrix}1&1\\1&1\end{smallmatrix}\right)+\frac{1}{N}\left(\begin{smallmatrix}0&0.3\\0.1&0\end{smallmatrix}\right)$; lower right: Coordination $\left(\begin{smallmatrix}1&1\\1&1\end{smallmatrix}\right)+\frac{1}{N}\left(\begin{smallmatrix}0.3&0\\0&0.1\end{smallmatrix}\right)$. In all cases, the leading and subleading eigenvalues of the core matrix are close to 1. }
\label{fig:wf}
 \end{figure}

%

\subsection{The minimum entropy}
\label{ssec:eigenvector}
The plots in Figure~\ref{fig:wf} bring an additional insight: the entropy curves for different initial conditions given by pure states do not cross as the system evolves through generations. This turns out to have unexpected consequences.

Indeed,  for an initial condition given by a pure state $\bp^\ini=\be_k$ and BGS entropy $E(\be_k)=-\log(v^{(1)}_ku^{(1)}_k)$. The minimum initial entropy is attained by a state supported in $\mathop{\mathrm{arg\, max}}_kv^{(1)}_ku^{(1)}_k$ --- when this set is a singleton, then the entropy minimizer distribution will be given by a pure state. In other words, the pure states corresponding to the maximum of the stationary distribution of the associated $Q$ process are the ones that minimise initial entropy in the system, provided such maximum is unique.

On the other hand, according to Lemma~\ref{lem:asymptotic_decay},  the entropy will be minimized when $n\to\infty$ at the minimum of $\langle \bu^{(2)},\bp\rangle^2=\left(u^{(2)}_k\right)^2$ (assuming $\lambda_2\ne0$, which is the generic situation). Considering an initial pure state, this will be achieved at $k$ such that $u^{(2)}_k$ is closest to zero.
	
Based on the numerical evidence provided by Figure~\ref{fig:wf}, the two observations above are linked and imply a relationship between these two conditions. We were unable to to provide a prove of this relationship, however Figure~\ref{fig:abplane} provides significant numerical evidence for this result. In particular, this suggest\added[id=facc]{s} that properties of associated entropy along the course of evolution may also bring further theoretical insights.

\begin{figure}
    \centering
    \includegraphics[width=\linewidth]{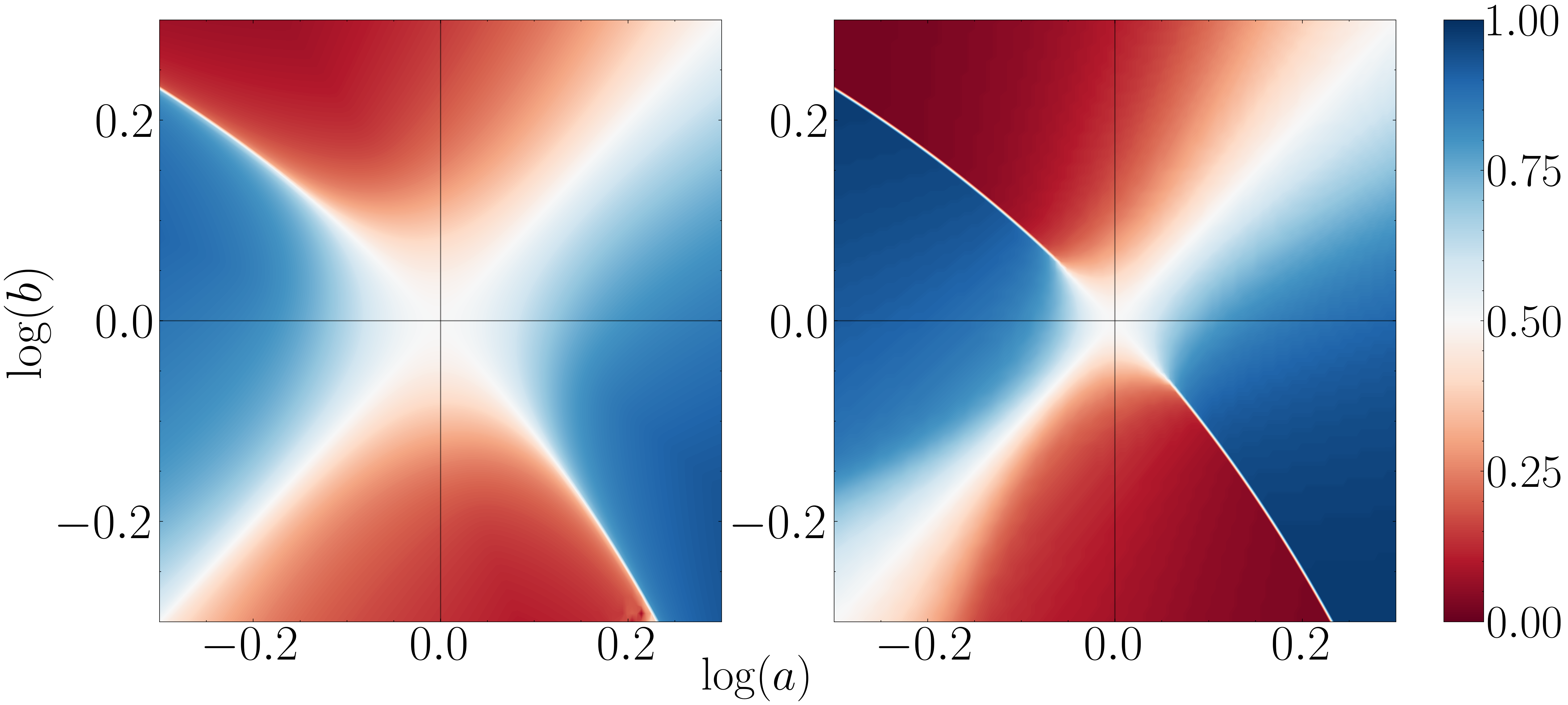}
    \caption[y]{Consider the Moran process with type selection probability vector $\tsp$ given by~\eqref{eq:tsp_fitness} and $\Psi^{(\A,\B)}$ given by game theory, with game matrix $\left(\begin{smallmatrix}1&a\\b&1\end{smallmatrix}\right)$. Left:  minimum entropy in the long run for a pure initial condition $\be_k$, $k=0,\dots,N$, given by  $\mathop{\mathrm{arg\,min}}_k|u^{(2)}_k|$, cf. equation~\eqref{eq:asymptotic_decay}. Right:  minimum BGS entropy for pure states $\be_k$, $k=0,\dots,N$ given by $\mathop{\mathrm{arg\, max}}_kv^{(1)}_ku^{(1)}_k$, cf. equation~\eqref{def:entropy}. Both values given as function of $\log a$ (horizontal axis) and $\log b$ (vertical axis), with $a,b\ge 0$. In both cases, the first quadrant represents coexistence games, the third, coordination games, and the second and fourth, domination games.}
    \label{fig:abplane}
\end{figure}

\subsection{Fundamental symmetries}
\label{sec:symmetries}

In fundamental physics, there are three fundamental discrete symmetries: time reversal symmetry (T-symmetry), charge conjugation symmetry (C-symmetry) and parity symmetry (P-symmetry) (see, e.g.,~\cite{hatfield1998quantum,sakurai2006advanced} and references therein). While there is no clear translation of these concepts within the framework of population genetics, we can use them as inspiration to explore fundamental symmetries in population genetics models. In this vein, we introduce the following definition:

 Let $\Phi_t(\chi;\mathbb{F},\mathbf{\Psi})$ the time advancing map of a given model in population genetics, where $\chi$ is the initial state, $\bF\bydef (x,1-x)$ is the vector of population state, and $\mathbf{\Psi}$ is the corresponding vector of fitnesses functions.  In the specific case of populations with only two types with constant population --- the case being considered in this work --- only $\mathbb{Z}_2$ symmetries are possible, and we use superscripts for the corresponding involution --- i.e $[(x,1-x)]^{\mathrm{p}}\bydef(1-x,x)$ and $[(\psi_1,\psi_2)]^{\mathrm{c}}\bydef(\psi_2,\psi_1)$. Moreover, when we change the vector of focal types, we assume the initial state is modified accordingly --- and we write $\chi^{\mathrm{p}}$ for the corresponding modification.  With this notation, we introduce the following definitions:

\begin{enumerate}
 \item \textbf{``T-symmetry'' in physics; time symmetry in population dynamics}: $\Phi_t(\chi;\mathbb{F},\mathbf{\Psi})=\Phi_{-t}(\chi;\mathbb{F},\mathbf{\Psi})$;
 \item \textbf{``P-symmetry'' in physics; type (or \emph{population}) symmetry in population dynamics}: $\Phi_t(\chi;\mathbb{F},\mathbf{\Psi})=\Phi_t(\chi^{\mathrm{p}};\mathbb{F}^{\mathrm{p}},\mathbf{\Psi})$;
 \item \textbf{``C-symmetry'' in physics; fitness (or \emph{choice}) symmetry in population dynamics}: $\Phi_t(\chi;\mathbb{F},\mathbf{\Psi})=\Phi_t(\chi;\mathbb{F},\mathbf{\Psi}^{\mathrm{c}})$.
\end{enumerate}

It is clear that in the neutral case, all models possess C- and P-symmetries.

We now look at each model and discuss the corresponding symmetries:

\paragraph{Kimura} Kimura model is given by  
\[
\partial_tp=\frac{\kappa}{2}\partial_x^2\left(x(1-x)p\right)-\partial_x\left(x(1-x)\Delta\psi(x)p\right), \quad p(x,0)=p_0(x)\ ,
\]
where $\Delta\psi=\psi^{\A}-\psi^{\B}$. Being a parabolic equation, the Kimura model does not have T- symmetry. The effect of a parity transformation is to change the sign of the first order term (in $x$) in the equation, and that is also the effect of a C-symmetry, provided $\psi_1\not=\psi_2$. In this case $\Delta\psi^{\textrm{cp}}(x)=-\Delta\psi(1-x)$, and therefore $p(t,1-x)$ is a solution of the CP-symmetic model if and only if $p(t,x)$ is a solution of the original model. Hence, this model possesses CP- symmetry.

\paragraph{Moran} Note that the C-symmetry transformation implies that the type selection vector changes as following
\[
\tsp_i^{\textrm{c}}\bydef\frac{i\Psi^{(B)}(i)}{i\Psi^{(\B)}(i)+(N-i)\Psi^{(\A)}(i)}=1-s_{N-i}\ ,
\]
and, therefore $M^{\textrm{c}}_{ij}=M_{N-i,N-j}$. We conclude that if $\bp$ is a solution of the Moran process, $\bp^{\textrm{c}}\bydef\left(p_{N-i}\right)_{i=0,\dots,N}$ is a solution of the C-Symmetric Moran process. However, $\bp^{\textrm{c}}=\bp^\rp$, where $\bp^\rp$ is the P-symmetric state vector. This shows that Moran also possesses CP- symmetry.

\paragraph{Replicator (PDE version)} This model is given by 
\[
\partial_tp+\partial_x\left(x(1-x)\Delta\psi(x)p\right)=0,\quad p(x,0)=p_0(x).
\]
For this model, T-symmetry changes the sign of the time derivative and C-symmetry changes the sign of the first order $x$-derivative; therefore the model possesses TC-symmetry. Furthermore, P-symmetry implies the changes $\Psi^{\mathrm{p}}(x)=\Psi(1-x)$, $p^{\mathrm{p}}(t,x)=p(t,1-x)$, $\partial_x^{\mathrm{p}}=-\partial_x$, and hence, any pair of transformation leaves the corresponding equation invariant.

For the ODE version of the Replicator equation $\dot{X}=X(1-X)\Delta\Psi(X)$, $X(0)=X_0$, the same symmetries of the PDE version of the Replicator equation are valid.

Table~\ref{tab:sym} summarize all cases discussed above.

\begin{table}
	\centering
	\begin{tabular}{|l|c|c|c|c|c|c|c|}
		\hline
		Model& C-&P-&T-&CP-&CT-&PT-&CPT-\\\hline\hline
		Kimura& -& -&-&X&-&-&-\\\hline
		Moran&-&-&-&X&-&-&-\\\hline
		Replicator (PDE \& ODE)&-&-&-&X&X&X&-\\\hline
	\end{tabular}
\caption{ Symmetries for the various models in the generic case. \label{tab:sym}}
\end{table}

\section{Conclusisons and biological implications}

This work investigates the reversible and  irreversible features of a class of absorbing processes that are ubiquitous in population genetics. In particular, it explores the application of a family of entropies discussed in \cite{CMRS:21} as a  tool to characterise these features. This family turns out to be  relative entropies of the associated $Q$ process -- a process that was instrumental in the derivation of the gradient structure of entropy minimisation in \cite{CMRS:21}. In this vein, it is an attempt to  provide the first steps towards a mathematical foundation of entropy and second law in biology, with an emphasis on reducible processes --- for other approaches to entropy in evolution and biology see \cite{brooks1988evolution} and  the review by \cite{roach2020use}; see also \cite{lieb1999physics} for a mathematical formulation of entropy and second law in physics and \cite{Zhu2020} for an application on bacterial resistance.

As a preliminary result, Lemma~\ref{lem:prob_paths} already suggests the possibility of a  macroscopic second law acting at population level and it is also seems to be  compatible with path entropies --- see \cite{roach2020use} for a critique on the second law approach while  supporting  the use of path entropies.

Lemma~\ref{lem:decay} provides a general proof of decaying of this family of entropies, which include the BGS and Tsallis entropies as special cases. Long term decay rates were also obtained in Lemma~\ref{lem:asymptotic_decay} and the results in  Section~\ref{ssec:eigenvector} are a byproduct of the inherent  features this family of entropies display and these decaying asymptotic rates.

The results in Section~\ref{ssec:multiloci} discuss an adaptation of the results for a multiloci framework that suggest the information correlation between different loci is linked to the entropy associated to the system. In this sense, Tsallis entropies should be the appropriate entropy for subadditive systems, and this raises the possibility of characterisation of epistatic systems using $m$-Tsallis entropies, where $m$ should measure the degree of correlation between different \emph{loci}.

At a more speculative level, we single out two questions that we believe are amenable to be addressed  by the methods developed here:

\begin{enumerate}
	\item When the Shannon entropy is appropriately rescaled, it yelds a metric of diversity called eveness~\cite{jost2010relation}. Use of similar metrics using the entropies studied here might yield some insight in measuring the ability of a population to develop resistance to control methods --- for instance a bacterial population to develop antibiotic resistance. Indeed, given certain macroscopic features of the population, we should not only estimate the average properties at the individual level, but also we should estimate how precise are our measurements vis-a-vis the true state of the population. These diversity indicators might provide an estimate on the number of microstates that are compatible with a measured one, and thus to quantify the uncertainty in the outcome of a human intervention on this population,
	\item  Another interesting possibility is to understand sudden changes of the macrostate of the population. It is known~\cite{ChalubSouza:2017a} that the Moran process is incompatible with jumps in the evolution, but for multi-agent interactions, it is possible to have discontinuities in the Wright-Fisher evolution. This is possibly related to a central discussion in evolutionary biology, i.e., whether evolution is gradual and slow or essentially  composed by fast changes and long stasis periods. This points out  in the direction of expanding results from the present work to multi-player game theory, without assuming weak selection.
\end{enumerate}

\section*{Ackowledgements}

DCC thanks the support by the Programme New Talents in Quantum Technologies of the Gulbenkian Foundation (Portugal). FACCC was partially funded by project UID/MAT/00297/2019. FACC also acknowledge some discussion with Armando Neves (UFMG, Brazil). MOS  was partially financed by Coordena\c{c}\~ao de Aperfei\c{c}oamento de Pessoal de N\'ivel Superior - Brasil (CAPES) - Finance code 001 and by the CAPES PRINT program at UFF --- grant \# 88881.310210/2018-01. MOS was also partially financed by CNPq (grant \# 310293/2018-9) and by FAPERJ (grant \# E-26/210.440/2019). We also thank an anonymous referee and the managing editor for many valuable comments that helped us to improve the paper.
 

\end{document}